\documentclass[12pt,a4paper]{article}

\usepackage[english]{babel}
\usepackage{graphicx}
\usepackage{enumerate}
\usepackage{amsmath}
\usepackage{color, colortbl}

\newcommand{\vsu}{\vspace{+.1cm} }
\newcommand{\vsd}{\vspace{+.2cm} }

\newtheorem{theorem}{Theorem}[section]
\newtheorem{defin}[theorem]{Definition}
\newtheorem{lemma}[theorem]{Lemma}

\newtheorem{definition}[theorem]{Definition}

\newtheorem{remark}[theorem]{Remark}

\newtheorem{proposition}[theorem]{Proposition}

\newenvironment{proof}{\par\noindent {\bf Proof.} \rm}{\ ~~~$\fbox{}$}

\newcommand{\N}{\mbox{\rm I$\!$N}}
\newcommand{\Z}{\mbox{\rm \lower0.3pt\hbox{$\angle\!\!\!$}Z}}

\newcommand{\sub}[2]{#1{[#2]}}

\newcommand{\card}[1]{{{\rm Card}\left(#1\right)}}

\newcommand{\neww}[1]{\omega}
\newcommand{\newx}[1]{\chi}
\newcommand{\newy}[1]{\gamma}
\newcommand{\newp}[1]{\varrho}
\newcommand{\news}[1]{\varsigma}

\newcommand{\factors}[1]{{\rm Fcts} \left( #1 \right)}

\title{$3$-anti-power uniform morphisms}

\author{Francis Wlazinski
}

\begin{document}

\maketitle

\parindent=0cm
\parskip=0.15cm

 
 

\newcounter{comptnivun}
\setcounter{comptnivun}{1}
\newcounter{comptnivdeux}
\setcounter{comptnivdeux}{1}
\newcounter{comptnivtrois}
\setcounter{comptnivtrois}{1}
\newcounter{comptnivquatre}
\setcounter{comptnivquatre}{1}

\bibliographystyle{plain}


\begin{abstract}

Words whose every three successive factors of the same length are all different i.e. 
$3$-anti-power words are a natural extension of square-free words (every two successive factors of the same length are different).
We give a way to verify whether a uniform morphism preserves $3$-anti-power words (the image of a 
$3$-anti-power word is a $3$-anti-power word).

A consequence of the existence of such morphisms is the possibility of generating infinite $3$-anti-power words.

\end{abstract}

\section{\label{sectionPreliminaries} 
         Preliminaries}

Let us recall some basic notions of Combinatorics of words.

\subsection{\label{sectionWords}Words}

An \textit{alphabet} $A$ is a finite set of symbols called \textit{letters}.
A \textit{word} over $A$ is a finite sequence of letters
from $A$.
The \textit{empty word} $\varepsilon$ is the empty sequence of letters.
Equipped with the concatenation operation, the set $A^*$ of words over $A$
is a free monoid with $\varepsilon$ as neutral element and $A$ as
set of generators.
Since an alphabet with one element is limited interest to us, we always assume that the cardinality of considered alphabets is at least two.
Given a non-empty
word $u = a_1\ldots a_n$ with $a_i \in A$ for any integer $i$ from 1 to $n$, the \textit{length}
of $u$ denoted by $|u|$ is the integer $n$
that is the number of letters of $u$.
By convention, we have $|\varepsilon| = 0$.
We denote by $A^+$ the set of words of positive length over $A$, i.e.,
$A^+=A^* \setminus \{\varepsilon\}$.

An infinite word over $A$ is a map from $\N$ to $A$ that is an infinite
sequence of letters $a_1\ldots a_n\ldots$ with $a_i \in A$.
And $A^{\N}$ is the set of all infinite words over $A$.


A word $u$ is a \textit{factor} of a word $v$ if there exist
two (possibly empty) words $p$ and $s$ such that $v = p u s$.
We denote $\factors{v}$ the set of all factors of $v$.
If $u\in\factors{v}$, we also say that $v$ \textit{contains} the word $u$ (as a factor).
If $p = \varepsilon$, $u$ is a \textit{prefix} of $v$.
If $s = \varepsilon$, $u$ is a \textit{suffix} of $v$.
If $u \neq v$, $u$ is a \textit{proper} factor of $v$.
If $u$, $p$ and $s$ are non-empty, $u$ is an \textit{internal} factor of $v$.


Let $w$ be a non-empty word and let $i, j$ be two integers such that
$0 \leq i-1 \leq j \leq |w|$.
We denote by $\sub{w}{i..j}$ the factor of $w$ such that $|\sub{w}{i..j}|=j-i+1$ 
and $w = p \sub{w}{i..j} s$ for two words $s$ and $p$ verifying $|p| = i-1$.
When $j>i$, $\sub{w}{i..j}$ is simply the factor of $w$ that starts at 
the $i^{\rm th}$ letter of $w$ and ends at the $j^{\rm th}$.
Note that, when $j = i - 1$, we have $\sub{w}{i..j} = \varepsilon$.
When $i=j$, we also denote by $\sub{w}{i}$ the factor $\sub{w}{i..i}$ which is the
$i^{\rm th}$ letter of $w$.
In particular, $\sub{w}{1}$ and $\sub{w}{|w|}$ are respectively the first and the last
letter of $w$.

Powers of a word are defined inductively by $u^0 = \varepsilon$, and 
for any integer $n \geq 1$, $u^n = u u^{n-1}$.
Given an integer $k \geq 2$, since the case $\varepsilon^k$ is of little interest, 
we call a \textit{$k$-power}
any word $u^k$ with $u\neq \varepsilon$.

Given an integer $k \geq 2$, a word is \textit{$k$-power-free} if it does not contain any 
$k$-power as factor. A \textit{primitive} word is a word which is not a $k$-power of 
another word whatever the integer $k \geq 2$.

Given two integers $p>q>1$ and two words $x$ and $y$ with $xy\neq \varepsilon$, 
a word of the form $(xy)^{\alpha}x$ with $\alpha + \frac{|x|}{|xy|} = \frac{p}{q}$
is called \textit{a $\frac{p}{q}$-power}. For instance, 
the word anchorman is a $(1+\frac{2}{7}=)\frac{9}{7}$-power
and the word $abaabaa$ is a $(2+\frac{1}{3}=)\frac{7}{3}$-power.
In particular, a $\frac{3}{2}$-power is a word of the form $xyx$ with $|x|=|y| \neq \varepsilon$.
For instance, the word antman is a $\frac{3}{2}$-power.
A word is \textit{$\frac{p}{q}$-power-free} if it does not contain any 
$\ell$-power as factor with $\ell \geq \frac{p}{q}$.
The word $abcaba$ is not $\frac{3}{2}$-power-free. Indeed, it 
contains the word $abc \, ab$ which is a $\frac{5}{3}$-power.

Given an integer $k \geq 2$ and an integer $n \geq 1$, 
a \textit{$(k,n)$-anti-power sequence} or simply a 
\textit{$k$-anti-power}~\cite{Fic2019} is a concatenation of $k$ 
consecutive pairwise different words of the same length $n$.

For instance, if $A=\{a,b\}$, the words $u=aa\,ba\,bb\, ab$ and
$v=aba \,bab \, abb \, aaa$ are respectively $(4,2)$-anti-power 
and $(4,3)$-anti-power sequences.
But the prefixe $ab ab ab ab $ of $v$ is not a $4$-anti-power sequence:
it is even a $4$-power.

Given an integer $k \geq 1$, if $\card A = \alpha \geq 2$ then there exit
$\alpha^n$ different words in $A^*$ of length $n \geq 1$.
Among the words of length $k \times n$, there are $\alpha^n$ different
$k$-powers (of length $k \times n$) and $A^k_{\alpha^n}=\dfrac{(\alpha^n)\,!}{(\alpha^n-k)\,!}$ 
different $(k,n)$-anti-power sequences if  $\alpha^n  \geq k$ and $0$ otherwise.
It particulary means that there exists an integer $k_0$
such that there are no $(k',n)$-anti-power sequence over $A$ for any $k' \geq k_0$.

For any alphabet $A$ with $\card A = \alpha \geq 2$ 
and for any integer $k \geq 2$, there exists a smallest integer $p_0$
such that $\alpha^{p_0}  \geq k$. And, if $p \geq p_0$, the set of 
$(k,p)$-anti-power sequences is greater than the set of $k$-powers
of length $p \times k$.

A \textit{$2$-anti-power word} is simply a square-free word.
Given an integer $k \geq 3$, a word $w$ is a \textit{$k$-anti-power word} if it is
a $(k-1)$-anti-power word and if any factor
of $w$ of length $k \times \ell$ for every $1 \leq \ell \leq \left\lfloor\dfrac{|w|}{k} \right\rfloor$ is a $(k,\ell)$-anti-power sequence.
By this definition, a word of length $n$ with $2 \leq n <k$
is a $k$-anti-power word if and only if it is a $n$-anti-power word.


An \textit{infinite $k$-anti-power word} is an infinite word whose 
finite factors are all $k$-anti-power words.
Obviously, the first question is whether such a word exists.

\vsu

If $A=\{a,b\}$, the only $2$-anti-power words are $aba$, $bab$ and their factors.
But, for any $k \geq 3$, the only $k$-anti-power words are $a$, $b$, $ab$ and $ba$.

If $\card A \geq 3$, there exist infinite $2$-anti-power (square-free) words
\cite{Ber1995,Thu1906,Thu1912}.

If $k=3$ and $A=\{a,b,c\}$, the only $3$-anti-power words are $abcab$, the exchange 
of letters of this word and their factors.
Let us note that the word $abcab$ is not $\frac{3}{2}$-power-free.

A $\frac{3}{2}$-power-free word contain neither a factor of the form
$xyx$ with $|x|=|y|$, nor a factor of the form $xx$.
Thus a $\frac{3}{2}$-power-free word is a $3$-anti-power word (but the converse does not hold). 
Thus a Dejean's word~\cite{Dej1972,CurRam2011,RAO2011} over a four-letter alphabet,
which does not contain any $\ell$-power with $\ell>\frac{7}{5}$-power-free, is a $3$-anti-power word.

More generally, a non-$k$-anti-power word (among $k$ consecutive factors of the same length of 
this word, at least two of them are equal) contains at least one fractionnal $\ell$-power
with $\ell \geq \frac{k}{k-1}$. Therefore, when $k \geq 3$, a Dejean's word over a 
$(k+1)$-letter alphabet is a $k$-anti-power word.

\begin{remark}

If we had chosen not to add that a $k$-anti-power word
must be a $(k-1)$-anti-power word, we would have, for instance, that,
for $A=\{a;b;c\}$, the word $abcabcab$ would have been a $3$-anti-power word
but not a $2$-anti-power word.

More precisely, without the condition that a $k$-anti-power word $w$
must be a $(k-1)$-anti-power word, we only could say that all prefixes and all
suffixes of $w$ of length between $k-1$ and $\left\lfloor \dfrac{(k-1)|w|}{k} \right\rfloor$
are $(k-1)$-anti-power words.

For an infinite word, it does not change anything to add the condition that a $k$-anti-power 
word $w$ must be a $(k-1)$-anti-power word. Indeed, 
every factor of $w$ whose length is a multiple of $k-1$  can be extended to a factor 
whose length is a multiple of $k$.
Obviously, if these $k$ factors are different, the same holds for $k-1$ ones.

\end{remark}




\begin{lemma}\cite{Ker1986,Lec1985}\label{factint}
If a non-empty word $v$ is an internal factor of $vv$,
i.e., if there exist two non-empty words $x$ and $y$ such that $vv=xvy$,
then there exist a non-empty word $t$ and two integers $i,j \geq 1$
such that $x=t^i$, $y=t^j$, and $v=t^{i+j}$.
\end{lemma}

\subsection{\label{sectionMorphisms}Morphisms}

Let $A$ and $B$ be two alphabets.
A \textit{morphism} $f$ from $A^*$ to $B^*$ is a mapping
from $A^*$ to $B^*$ such that $f(uv) = f(u)f(v)$ for all words $u, v$ over $A$.
When $B$ has no importance, we say that $f$ is a morphism on
$A$ or that $f$ is defined on $A$. 

Given an integer $L \geq 1$, $f$ is \textit{$L$-uniform}
if $|f(a)| = L$ for every letter $a$ in $A$.
A morphism $f$ is \textit{uniform} if it is $L$-uniform for some integer $L \geq 1$.

Let $k \geq 2$ be an integer and Let $A$ and $B$ be two alphabets.
A morphism $f$  from $A^*$ to $B^*$  is \textit{$k$-anti-power}  if and only
if $f(w)$ is a $k$-anti-power word over B for every $k$-anti-power word $w$ over $A$
For instance, the \textit{identity endomorphism} $Id$
($\forall a \in A$, $Id(a) = a$) 
is a $k$-anti-power morphism.
In particular, a $2$-anti-power morphism is a square-free morphism.
%
%
These last morphisms have been characterized in~\cite{Cro1982}.

We say that a morphism is \textit{non-erasing} if, for all letters $a \in A$,
$f(a) \neq \varepsilon$.
A $k$-anti-power morphism, as every square-free morphism, is necessarily non-erasing.


A morphism on $A$ is called \textit{prefix} (resp. \textit{suffix})
if, for all different letters $a$ and $b$ in $A$, 
the word $f(a)$ is not a prefix (resp. not a suffix) of $f(b)$.
A prefix (resp. suffix) morphism is non-erasing.
A morphism is \textit{bifix} if it is prefix and suffix.


Proofs of the two following lemmas are left to the reader.


\begin{lemma}
\label{pref}
Let $f$ be a prefix morphism on an alphabet $A$,
let $u$ and $v$ be words over $A$,
and let $a$ and $b$ be letters in $A$.
Furthermore,
let $p_1$ (resp. $p_2$) be a prefix of $f(a)$ (resp. of $f(b)$).
If $(p_1;p_2) \neq (\varepsilon;f(b))$ and if $(p_1;p_2) \neq (f(a);\varepsilon)$ 
then
the equality $f(u)p_1=f(v)p_2$ implies $u=v$ and $p_1=p_2$.
\end{lemma}

\begin{lemma}
\label{suff}
Let $f$ be a suffix morphism on an alphabet $A$,
let $u$ and $v$ be words over $A$, 
and let $a$ and $b$ be letters in $A$.
Furthermore,
let $s_1$ (resp. $s_2$) be a suffix of $f(a)$ (resp. of $f(b)$).
If $(s_1;s_2) \neq (\varepsilon;f(b))$ and if $(s_1;s_2) \neq (f(a);\varepsilon)$ 
then
the equality $s_1f(u)=s_2f(v)$ implies $u=v$ and $s_1=s_2$.

\end{lemma}

Taking $p_1=p_2=\varepsilon$ (resp. $s_1=s_2=\varepsilon$) in Lemma~\ref{pref}
(resp Lemma~\ref{suff}),
we get that a prefix (resp. suffix) morphism is injective.

\begin{definition}

A morphism $f$ from $A^*$ to $B^*$ is a \textit{ps-morphism} (Ker\"anen \cite{Ker1986} called 
$f$ a ps-code)
if and only if the equalities \\ \centerline{$f(a) = ps$, $f(b) = ps'$ and
$f(c) = p's$} with $a,b,c \in A$ (possibly $c = b$) and $p$, $s$, $p'$,
$s' \in B^*$ imply $b = a$ or $c = a$.

\end{definition}

Obviously, taking $c=b$, and $s=\varepsilon$ in a first time
and $p=\varepsilon$ in a second time, we obtain that
a ps-morphism is a bifix morphism.

\begin{lemma}{\rm \cite{Ker1986,Lec1985}}
\label{lemmeSPKPS}
If $f$ is not a ps-morphism then $f$ is not a $k$-power-free morphism for every integer $k \geq 2$.
\end{lemma}


\begin{remark}
It means that a $2$-anti-power morphism is a ps-morphism.
\end{remark}

\begin{proposition}\label{pro2} 

Let $A$ and $B$ be two alphabets with $\card A \geq 2$
and let $f$ be a uniform morphism from $A^*$ to $B^*$.
If there exist five letters $a$, $b$, $c$, $d$ and $x$ (possibly equal)
and four words $p$, $s$, $\pi$ and $\sigma$ such that
$s$ is a suffix of $f(a)$, $p$ is a prefix of $f(b)$,
$\sigma$ is a non-empty suffix of $f(c)$, $\pi$ is a non-empty prefix of $f(d)$,
and $sp=\sigma f(x) \pi$ then $f$ is not a square-free morphism.

\end{proposition}

\begin{proof}

Since $|sp|>|\sigma f(x)|$ and $|sp|>|f(x)\pi|$,
we get $|s|> |\sigma|$ and $|p|>|\pi|$.

Let $s'$ be the non empty prefix of $f(x)$ such that $s= \sigma s'$
and
let $p'$ be the non empty suffix of $f(x)$ such that $p= p' \pi$.

If $x=a$ and $x=b$, then $f(x)$ is an internal factor of $f(xx)$.
By Lemma~\ref{factint}, $f(x)$ is not primitive, i.e.,
$f$ is not a square-free morphism.

Therefore, at least one of the word $ax$ or $bx$ is not a square. 
But $f(ax)$ contains the square $s's'$ 
and $f(xb)$ contains the square $p'p'$, i.e.,
 $f$ is not a square-free morphism.
\end{proof}

\begin{proposition}\label{pro1} 

Let $A$ and $B$ be two alphabets with $\card A \geq 3$
and let $f$ be a $L$-uniform morphism from $A^*$ to $B^*$.
If $L$ is an even number then $f$ is not a $3$-anti-power morphism.

\end{proposition}

\begin{proof}

Let $a$, $b$ and $c$ be three different letters in $A$.
Let $A_1$ and $A_2$ be the words such that $f(a)=A_1A_2$ with $|A_1|=|A_2|$.
Then $f(abcab)$
contains the non-$3$-anti-power sequence $[A_2f(b)] \, f(c) A_1 \, [A_2f(b)]$ with $abcab$ a 
$3$-anti-power word.
\end{proof}

\vsd
A morphism $f$ on $A$ is \textit{$k$-anti-power up to $\ell$} ($k,\ell \geq 2$) if and only
if $f(w)$ is a $k$-anti-power word for every $k$-anti-power word $w$ over $A$
of length at most $\ell$.

\begin{proposition}\label{mainresult} 
Let $A$ and $B$ be two alphabets with $\card A \geq 3$
and let $f$ be a square-free $L$-uniform morphism from $A^*$ to $B^*$.

Then $f$ is a $3$-anti-power morphism if and only if
it is a $3$-anti-power morphism up to $5$.

\end{proposition}

\begin{proof}

If $L$ is an even number, the image of the word $abcab$ 
of length $5$ shows that $f$ is not a 3-anti-power morphism
(see the proof of Proposition~\ref{pro1}). It ends the proof.
So we may assume that $L$ is odd.

By definition of $3$-anti-power
morphisms, we only have to prove the "if" part of Proposition~\ref{mainresult}.

By Lemma~\ref{lemmeSPKPS}, $f$ (square-free) is a ps-morphism and so injective.

By contradiction, we assume that a shortest $3$-anti-power word 
$w$ (not necessarily unique) such that $f(w)$
contains a non-$3$-anti-power satisfies $|w| \geq 6$.
We will show that this assumption leads to contradictions.

Since the length of $w$ is minimal, we may assume that there exist five words
$p$, $s$, $U_1$, $U_2$ and $U_3$ such that $f(w)=p U_1U_2U_3 s$
where $p$ is a prefix of $f(\sub{w}{1})$ different from
$f(\sub{w}{1})$ and $s$ is a suffix of $f(\sub{w}{|w|})$
different from $f(\sub{w}{|w|})$. Moreover, the words $U_1$, $U_2$, $U_3$
have the same length $\Lambda (\geq 1)$ and two of them are equal.

If $U_1=U_2$ or if $U_2=U_3$, it means that $f(w)$ contains a square with
$w$ a $3$-anti-power word so a square-free word. 
That is $f$ is not a square-free morphisms: a contradiction with the definition of $f$.
The only remaining case is $U_1=U_3$.
To simplify notations, we denote by $U$ the words $U_1$ and $U_3$
and by $V$ the word $U_2$.

Let $i_2$ be the shortest integer such that $p U$
is the prefix of $f(\sub{w}{1..i_2})$
and let $i_3$ be the shortest integer such that $p UV$
is the prefix of $f(\sub{w}{1..i_3})$.

If $i_2=1$ then $\Lambda \leq |pU|  \leq |f(\sub{w}{1})|=L$ and $|p UVU s| < 4L$.
This is impossible since $|f(w)| \geq 6L$.

On a the same way, by a length criterion, 
the cases $i_2=i_3$ and $i_3=|w|$ are impossible.

If we denote $x=\sub{w}{2..i_2-1}$,
$y=\sub{w}{i_2+1..i_3-1}$, $z=\sub{w}{i_3+1..|w|-1}$,
$a_1=\sub{w}{1}$, $a_2=\sub{w}{i_2}$, $a_3=\sub{w}{i_3}$
and $a_4=\sub{w}{|w|}$ then $w=a_1xa_2ya_3za_4$
with $|w|= 4 + |x| +|y| + |z|$. It implies that
$|x| +|y| + |z| \geq 2$.

Moreover, there exists some words $p_i$ and $s_i$ ($1 \leq i \leq 4)$
such that $f(a_i)=p_is_i$ with
$p_1=p$, $s_4=s$.
By definition, the words $s_1$, $p_2$, $p_3$ and $p_4$ are non empty.

In other words, we have $U=s_1f(x)p_2=s_3f(z)p_4$ and $V=s_2f(y)p_3$.

Let $x_{j_1}$ and $x_{j_2}$ be two different words of $\{x,y,z\}$, we have
$||x_{j_1}|-|x_{j_2}|| \leq 1$. 
Indeed, in the contrary, for instance, if $|x_{j_1}| \geq |x_{j_2}| + 2$,
we get that both
$\Lambda > |f(x_{j_1})| \geq |f(x_{j_2})| + 2 L$
and $\Lambda  < |f(x_{j_2})| + 2 L$:
this is impossible.
Since $|x| +|y| + |z| \geq 2$, it also implies that at least
two of the words $x,y$ and $z$ are non empty.


$\bullet$ {\textit{Case 1 : $|s_1|=|s_3|$}}

Since $2 \times |U|=|U|+|V|=|s_1f(x)p_2|+|s_2f(y)p_3|=(|x|+|y|+2) \times L$,
we get that $|x|+|y|$ is even, i.e., $|x|=|y|$.

From the equality $s_1f(x)p_2=s_3f(z)p_4(=U)$, we get $s_3=s_1 (\neq \varepsilon)$.
By Lemma~\ref{pref}, it also implies $z=x$ and $p_2=p_4 (\neq \varepsilon)$.

In particular, since $|x|=|y|$ and since $w=a_1xa_2ya_3xa_4$ is a $3$-anti-power word, 
we have $a_1 \neq a_3$ and $a_2 \neq a_4$.

Since $|U|=|V|$, we get $|s_1|+|p_2|=|s_2|+|p_3|$.
Since $|s_1|=|s_3|$, we get $|p_1|=|p_3|$ and
$2|p_1|=L+|p_1|-|s_1|=L+|p_3|-|s_1|=L+|p_2|-|s_2|=2|p_2|$, i.e., 
$|p_1|=|p_2|$.

In a same way, since $|p_2|=|p_4|$, we get $|s_2|=|s_4|=|s_3|$.

If $a_1=a_2$ then $p_1=p_2(=p_4)$, $f(a_1)=p_1s_1$, $f(a_3)=p_3s_3=p_3s_1$
and $f(a_4)=p_4s_4=p_1s_4$.
It means that $f(a_3 a_1 a_4)$ contains $(s_1p_1)^2$
with $a_3 a_1 a_4$ square-free since $a_1 \neq a_3$ and 
$a_1=a_2 \neq a_4$: a contradiction with the hypothesis that $f$ is a
square-free morphism.

In the same way, if $a_3=a_4$, we get that 
$f(a_1 a_4 a_2)$ contains $(s_4p_4)^2$
with $a_1 a_4 a_2$ square-free.
And, if $a_2=a_3$, we get that
$f(a_1 a_2 a_4)$ contains $(s_2p_2)^2$
with $a_1 a_2 a_4$ square-free.
In theses both cases, we again get a contradiction 
with the hypothesis that $f$ is a square-free morphism.

Thus $a_1$, $a_2$ and $a_3$ are three different letters
and $a_2$, $a_3$ and $a_4$ are also three different letters.
It means that $a_1a_2a_3a_4$ is a $3$-anti-power word.
But $f(a_1a_2a_3a_4)$
contains the non-$3$-anti-power sequence
$s_1 p_2 \, s_2  p_3 \, s_3 p_4=s_1 p_2 \, s_2  p_3 \, s_1 p_2$:
a contradiction with the minimality of $|w|$.

$\bullet$ {\textit{Case 2 : $s_1=f(a_1)$ and $s_3=\varepsilon$}}

We have $|f(a_1)f(x)p_2|=|s_1f(x)p_2|=|U|=|V|=|s_2f(y)p_3|=|s_2f(y)f(a_3)|$, i.e.,
 $|p_2|=|s_2|$: this contradicts the fact that $L$ is odd.

$\bullet$ {\textit{Case 3 : $s_1 \neq s_3$}}

%
Since, at least two of the words $x,y$ and $z$ are non empty, we have
$x \neq \varepsilon$ or $z \neq \varepsilon$.

Since the equality $s_1f(x)p_2=s_3f(z)p_4$ is symetric, without loss of generality, we may assume
that $|s_1|<|s_3|$.
In this case, we necessarily have $x \neq \varepsilon$. Let $\chi$ be the first
letter of $x$ and let $x'$ be the word such that $x=\chi x'$.
If $z=\varepsilon$, let $P=p_4$ and if $z \neq \varepsilon$
let $P=f(\gamma)$ where $\gamma$ is the first letter of $z$.
In particular, we have $P$ non-empty.
Let $\pi$ be the (non empty) prefix of $f(x')p_2$ such that 
$s_1f(\chi)\pi=s_3P$.
By proposition~\ref{pro2}, this last equation implies that 
$f$ is not a square-free morphism: a final contradiction.
\end{proof}

\section{An example}

As stated in the first section, Dejean's words are anti-power words. But we can 
build $3$-power-free words without using fractionnal powers.

According to my computer, the following morphism $h$ is a $3$-anti-power morphism
(but I do not really trust my programming skills).

\begin{center}
$\begin{array}{ll}
h: & \{a;b;c;d;e\}^* \rightarrow \{a;b;c;d;e\}^* \\
& a \mapsto abceacd \\
& b \mapsto abecaed \\
& c \mapsto acbaecd\\
& d \mapsto acbeabd\\
& e \mapsto acebced\\
\end{array}$
\end{center}

The word $h^{\omega}(a)= \lim_{n \rightarrow +\infty} h^n(a)= abceacd \, abecaed \, acbaecd \, acebced \, abceacd \, acbaecd \, acbeabd \, $ $abceacd \, abecaed \, acebced \, acbaecd \, abceacd \, ...$ generated by $h$ is thus an infinite $3$-anti-power word.

Let us remark that $h^{\omega}(abcac)$ is also a  $3$-anti-power word.
But it contains an infinite number of factors that are $\dfrac{5}{3}$-powers.

\vsd

\textbf{Acknowledgment.}

I would like to thank James D. Currie for his comments on the first version of this paper.

\bibliography{bf.bib}

\end{document}